\documentclass[unicode,runningheads,envcountsame,envcountsect,a4paper,11pt]{llncs}

\usepackage{fullpage}

\usepackage{orcidlink}
\renewcommand{\orcidID}[1]{\orcidlink{#1}}

\usepackage[T1]{fontenc}

\usepackage{graphicx}

\usepackage{amsmath}
\usepackage{amssymb}
\usepackage{tcolorbox}
\usepackage{mathtools}

\usepackage{esvect}
\newcommand{\ori}{\vv} 

\usepackage[capitalise,nameinlink]{cleveref}
\spnewtheorem{myclaim}[theorem]{Claim}{\itshape}{\upshape}
\crefname{myclaim}{Claim}{Claims}

\spnewtheorem{observation}[theorem]{Observation}{\bfseries}{\upshape}
\crefname{observation}{Observation}{Observations}

\usepackage{thmtools}
\usepackage{thm-restate}

\usepackage{xcolor}
\definecolor{darkred}{rgb}{0.7,0,0}

\newcommand{\gpara}{\mathsf}
\DeclareMathOperator{\bn}{\gpara{b}} 
\DeclareMathOperator{\obn}{\gpara{B}} 
\DeclareMathOperator{\cc}{\theta} 
\DeclareMathOperator{\mn}{\mu} 

\DeclareMathOperator{\odn}{\gpara{DOM}} 

\DeclareMathOperator{\cvd}{\gpara{cvd}} 
\DeclareMathOperator{\vc}{\gpara{vc}} 

\newcommand{\KE}{K\H{o}nig--Egerv\'{a}ry}

\let\doendproof\endproof
\renewcommand\endproof{~\hfill$\qed$\doendproof}

\newenvironment{subproof}[1][\proofname]{%
  \renewcommand\endproof{~\hfill$\lozenge$\doendproof}
  \begin{proof}[#1]%
}{%
  \end{proof}%
}

\begin{document}
%
\title{Orientable Burning Number of Graphs\thanks{%
Partially supported
by JSPS KAKENHI Grant Numbers
JP18H04091, 
JP20H05793, 
JP21K11752, 
JP22H00513, 
JP23KJ1066. 
}}
%
%
\author{
Julien Courtiel\inst{1}\orcidID{0000-0002-3441-2818} \and
Paul Dorbec\inst{1}\orcidID{0009-0007-1179-6082} \and
Tatsuya Gima\inst{2,3}\orcidID{0000-0003-2815-5699} \and
Romain Lecoq\inst{1} \and
Yota Otachi\inst{2}\orcidID{0000-0002-0087-853X}
}

\authorrunning{J. Courtiel et al.}
%
\institute{%
Normandie Univ, UNICAEN, ENSICAEN, CNRS, GREYC, 14000 Caen, France \email{julien.courtiel@unicaen.fr},
 \email{paul.dorbec@unicaen.fr},
 \email{romain.lecoq@unicaen.fr}
\and
Nagoya University, Nagoya, Japan.
 \email{gima@nagoya-u.jp},
 \email{otachi@nagoya-u.jp}
\and
JSPS Research Fellow
}
\maketitle              
\begin{abstract}
In this paper, we introduce the problem of finding an orientation of a given undirected graph that maximizes the burning number of the resulting directed graph.
We show that the problem is polynomial-time solvable on {\KE} graphs (and thus on bipartite graphs)
and that an almost optimal solution can be computed in polynomial time for perfect graphs.
On the other hand, we show that the problem is NP-hard in general and W[1]-hard parameterized by the target burning number.
The hardness results are complemented by several fixed-parameter tractable results parameterized by structural parameters.
Our main result in this direction shows that the problem is fixed-parameter tractable parameterized by cluster vertex deletion number plus clique number
(and thus also by vertex cover number).

\end{abstract}


\section{Introduction}
The \emph{burning number} of a directed or undirected graph $G$, denoted $\bn(G)$,
is the minimum number of steps for burning all vertices of $G$ in the following way:
in each step, we pick one vertex and burn it; and then
between any two consecutive steps, the fire spreads to the neighbors (to the out-neighbors, in the directed setting) of the already burnt vertices.
In other words, $\bn(G)$ is the minimum integer $b$ such that there exists a sequence $\langle w_{0}, \dots, w_{b-1}\rangle$ of vertices
such that for each vertex $v$ of $G$, there exists $i$ ($0\le i \le b-1$) such that the distance from $w_{i}$ to $v$ is at most $i$.
Note that each $w_{i}$ corresponds to the vertex that we picked in the $(b-i)$th step.

The concept of burning number is introduced by Bonato, Janssen, and Roshanbin~\cite{BonatoJR14,BonatoJR16} as a model of information spreading,
while the same concept was studied already in 1992 by Alon~\cite{Alon92}.
The central question studied so far on this topic is the so-called \emph{burning number conjecture},
which is about the worst case for a burning process and
states that $\bn(G) \le \lceil \sqrt{n} \rceil$ for every connected undirected graph with $n$ vertices.
The conjecture has been studied intensively but it is still open (see \cite{Bonato20} and the references therein).
Recently, it has been announced that the conjecture holds asymptotically, that is, $\bn(G) \le (1 + o(1)) \sqrt{n}$~\cite{NorinT22arxiv}.
For the directed case, the worst cases are completely understood in both weakly and strongly connected settings~\cite{Janssen20}.
Since the problem of computing the burning number is hard~\cite{BessyBJRR17,LiuHH20a,MondalRPR22},
several approximation algorithms~\cite{LieskovskyS22,LieskovskySF23,Martinsson23arxiv}
and parameterized algorithms~\cite{AshokDKSTV23,KareR19,KobayashiO22} are studied.

In this paper, we investigate the worst case for a directed graph in the setting where we only know the underlying undirected graph.
That is, given an undirected graph, which is assumed to be the underlying graph of a directed graph,
we want to know how bad the original directed graph can be in terms of burning number.
This concept is represented by the following new graph parameter:
the \emph{orientable burning number} of an undirected graph $G$, denoted $\obn(G)$, is the maximum burning number over all orientations of $G$;
that is,
\[
  \obn(G) = \max_{\text{orientation $\ori{G}$ of $G$}} \bn(\ori{G}).
\]
Recall that an orientation $\ori{G}$ of an undirected graph $G$ is a directed graph that gives exactly one direction to each edge of $G$.
Now the main problem studied in this paper is formalized as follows.
\begin{tcolorbox}
\begin{description}
  \item[Problem:] \textsc{Orientable Burning Number} (OBN)
  \item[Input:] An undirected graph $G = (V, E)$ and an integer $b$.
  \item[Question:] Is $\obn(G) \ge b$?
\end{description}
\end{tcolorbox}

In the setting of information spreading applications,
this new problem can be seen as the one determining directions of links in a given underlying network structure
to make the spread of something bad as slow as possible.
Note that the dual problem of minimizing the burning number by an orientation
is equivalent to the original graph burning problem on undirected graphs since
$\bn(G) = \min_{\text{orientation $\ori{G}$ of $G$}} \bn(\ori{G})$.
To see this equality, observe that each edge is used at most once and only in one direction to spread the fire.

See \cref{fig:star} for an example on the star graph $K_{1,n}$.
\begin{figure}[htb]
\centering
\includegraphics{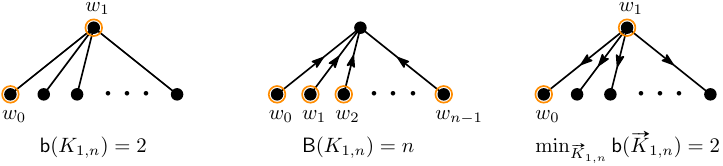}
\caption{The star graph $K_{1,n}$ with $n \ge 2$.}
\label{fig:star}
\end{figure}

\subsection{Our results}

We first present, in \cref{sec:bounds}, several lower and upper bounds
connecting the orientable burning number of a graph with other parameters such as the independence number.
In particular, for perfect graphs, we present almost tight lower and upper bounds that differ only by~$2$
and can be computed in polynomial time.
We also consider {\KE} graphs, which generalize bipartite graphs.
Although our bounds for them are not exact,
we show that the orientable burning number of a {\KE} graph can be computed in polynomial time (see \cref{sec:KE-graph}).

Next we consider the computational intractability of OBN in \cref{sec:hardness}.
We first show that OBN is W[1]-hard parameterized by the target burning number~$b$.
Although the proof of this result implies the NP-hardness of OBN for general graphs as well,
we present another NP-hardness proof that can be applied to restricted graph classes that satisfy a couple of conditions.
For example, this shows that OBN is NP-hard on planar graphs of maximum degree~$3$.

To cope with the hardness of OBN, we study structural parameterizations in \cref{sec:parameterizations}.
We first observe that some sparseness parameters (e.g., treewidth) combined with $b$ make the problem fixed-parameter tractable.
The main question there is the tractability of structural parameterizations \emph{not} combined with~$b$.
We show that OBN parameterized by cluster vertex deletion number plus clique number is fixed-parameter tractable.
As a corollary to this result, we can see that OBN parameterized by vertex cover number is fixed-parameter tractable.
We believe that the techniques used there would be useful for studying other structural parameterizations of OBN as well.

\subsection{Related problems}

Although the problem studied in this paper is new,
the concept of orientable number has long history in the settings of some classical graph problems.

The most relevant is the orientable domination number.
The \emph{orientable domination number} of an undirected graph $G$, denoted $\odn(G)$, is the maximum domination number over all orientations $\ori{G}$ of $G$.
That is,
\[
 \odn(G) = \max_{\text{orientation $\ori{G}$ of $G$}} \gamma(\ori{G}),
\]
where $\gamma(\ori{G})$ is the size of a minimum dominating set of the directed graph $\ori{G}$.\footnote{%
In a directed graph, a vertex dominates itself and its out-neighbors.}
Erd\H{o}s~\cite{Erdos63tournament} initiated (under a different formulation) the study of orientable domination number
by showing that $\odn(K_{n}) \simeq \log_{2} n$, where $K_{n}$ is the complete graph on $n$ vertices.
Later, the concept of orientable domination number is explicitly introduced by Chartrand et al.~\cite{ChartrandVY96}.
We can show that orientable domination number (plus~$1$) is an upper bound of orientable burning number
(see \cref{obs:obn_le_ods+1}).

There are two other well-studied problems.
One is to find an orientation that minimizes the length of a longest path,
which is equivalent to the graph coloring problem by the the Gallai--Hasse--Roy--Vitaver theorem~\cite{Gallai68,Hasse65,Roy67,Vitaver62}. %
The other one is to find a strong orientation that minimizes or maximizes the diameter.
It is NP-complete to decide if a graph admits a strong orientation with diameter~2~\cite{ChvatalT78}
and the maximum diameter of a strong orientation is equal to the length of a longest path in the underlying 2-edge connected  graph~\cite{Gutin94}.


\section{Preliminaries}
\label{sec:pre}

\paragraph{Terms in graph burning.}
Let $D = (V,A)$ be a directed graph.
By $N_{\ell,D}^{+}[v]$, we denote the set of vertices with distance at most $\ell$ from $v$ in $D$.
We often omit $D$ in the subscript and write $N^{+}_{\ell}[v]$ instead  when it is clear from the context.
A \emph{burning sequence} of $D$ with length $b$ is
a sequence $\langle w_{0}, w_{1}, \dots, w_{b-1}\rangle \in V^{b}$
such that $\bigcup_{0 \le i \le b-1} N_{i}^{+}[w_{i}] = V$.
Note that the burning number of $D$ is the minimum integer $b$ such that $D$ has a burning sequence of length $b$.
We call the $i$th vertex $w_{i}$ in a burning sequence the \emph{fire} of \emph{radius} $i$
and say that $w_{i}$ \emph{burns} the vertices in $N_{i}^{+}[w_{i}]$.

\paragraph{Some basic graph terms.}
Let $G = (V,E)$ be a graph.
The \emph{complement} of $G$, denoted $\overline{G} = (V, \overline{E})$, is the graph
in which two vertices are adjacent if and only if they are not adjacent in $G$.
For $S \subseteq V$, let $G[S]$ denote the subgraph of $G$ induced by $S$.
For $S \subseteq V$, let $G - S = G[V \setminus S]$.
A vertex set $S \subseteq V$ is an \emph{independent set} in $G$ if $G[S]$ contains no edge.
The \emph{independence number} of $G$, denoted $\alpha(G)$, is the maximum size of an independent set in $G$.
The \emph{chromatic number} of $G$, denoted $\chi(G)$, is the minimum integer $c$ such that
the vertices of $G$ can be colored with $c$ colors in such a way that no two adjacent vertices have the same color.
In other words, $\chi(G)$ is the minimum integer $c$ such that $V$ can be partitioned into $c$ independent sets.
A vertex set $S \subseteq V$ is a \emph{clique} in $G$ if $G[S]$ is a complete graph.
The \emph{clique number} of $G$, denoted $\omega(G)$, is the maximum size of a clique in $G$.
The \emph{clique cover number} of $G$, denoted $\cc(G)$, is the minimum integer $t$ such that $V$ can be partitioned into $t$ cliques.
An edge set $M \subseteq E$ is a \emph{matching} in $G$ if no two edges in $M$ share an endpoint.
The \emph{matching number} of $G$, denoted $\mn(G)$, is the maximum size of a matching in $G$.
Note that $\alpha(G) = \omega(\overline{G})$ and $\chi(G) = \cc(\overline{G})$ hold for every graph $G$.

A graph $G = (V,E)$ is a \emph{perfect graph} if $\omega(G[S]) = \chi(G[S])$ holds for all $S \subseteq V$.
Equivalently, $G$ is a perfect graph if $\alpha(G[S]) = \cc(G[S])$ holds for all $S \subseteq V$
since the class of perfect graphs is closed under taking complements~\cite{Lovasz72}.
The class of perfect graphs contains several well-studied classes of graphs such as bipartite graphs and chordal graphs (see, e.g., \cite{Bonomo-Braberman20}).
A graph $G = (V,E)$ is a \emph{{\KE} graph} if $\alpha(G) = |V| - \mn(G)$.
It is known that every bipartite graph is a {\KE} graph~\cite{Egervary1931,Konig1931}.

\paragraph{Structural parameters of graphs.}
Let $G = (V,E)$ be a graph.
A \emph{vertex cover} of $G$ is a set $S \subseteq V$ such that $G - S$ has no edge.
The \emph{vertex cover number} of $G$, denoted $\vc(G)$, is the minimum size of a vertex cover of $G$.
A \emph{cluster vertex deletion set} of $G$ is a set $S \subseteq V$ such that each connected component of $G - S$ is a complete graph.
The \emph{cluster vertex deletion number} of $G$, denoted $\cvd(G)$, is the minimum size of a cluster vertex deletion set of $G$.
In \cref{sec:parameterizations}, we use $\cvd(G) + \omega(G)$ as a parameter.
This combined parameter can be related to other parameters as follows.
\begin{observation}
\label{obs:paras_ineq}
For every graph $G$, the following inequalities hold:
\[
  \chi(G) \le \cvd(G) + \omega(G) \le 2 \vc(G) + 1.
\]
\end{observation}
\begin{proof}
Let $k = \cvd(G) + \omega(G)$.
To see that $\chi(G) \le k$,
one can construct a $k$-coloring by first coloring a cluster vertex deletion set $S$ of size $\cvd(G)$ with $\cvd(G)$ colors,
and then coloring each connected component of $G-S$ independently using at most $\omega(G)$ colors not used in $S$.
For $k \le 2 \vc(G) + 1$, observe that a vertex cover is a cluster vertex deletion set
and that the clique number cannot be larger than the vertex cover number plus~$1$.
\end{proof}
We can see that $\cvd(G) + \omega(G)$ is an upper bound of vertex integrity,
and thus of treedepth, pathwidth, treewidth, and clique-width.
We are not going to define these parameters as we do not explicitly use them in this paper.
For their definitions, see \cite{GimaHKKO22,HlinenyOSG08}.

We assume that the readers are familiar with the theory of parameterized algorithms.
(For standard concepts, see \cite{CyganFKLMPPS15,DowneyF13,FlumG06,Niedermeier06}.)
Recall that a parameterized problem with input size $n$ and parameter $k$ is \emph{fixed-parameter tractable}
if there is a computable function $f$ and a constant $c$ such that the problem can be solved in time $O(f(k) \cdot n^{c})$,
while being W[1]-hard means that the problem is unlikely to be fixed-parameter tractable.


\section{General lower and upper bounds}
\label{sec:bounds}

In this section, we present lower and upper bounds of orientable burning number,
which are useful in presenting algorithmic and computational results in the rest of the paper.
We believe that the bounds would be of independent interest as well.

We start with a simple observation that orientable burning number
is bounded from above by orientable domination number plus~$1$.
\begin{observation}
\label{obs:obn_le_ods+1}
$\obn(G) \le \odn(G) + 1$ for every graph $G$.
\end{observation}
\begin{proof}
Let $\ori{G}$ be an orientation of $G$. Let $\{w_{1}, \dots, w_{\gamma}\}$ be a minimum dominating set of $\ori{G}$.
By arbitrarily setting $w_{0}$, we construct a sequence $\sigma = \langle w_{0}, w_{1}, \dots, w_{\gamma} \rangle$ of length $\gamma+1$.
Since $\{w_{1}, \dots, w_{\gamma}\}$ is a dominating set, $\sigma$ is a burning sequence of $\ori{G}$.
Therefore, $\obn(G) = \max_{\ori{G}} \bn(\ori{G}) \le \max_{\ori{G}} \gamma(\ori{G}) + 1 = \odn(G) + 1$.
\end{proof}
Since $\odn(G) \in O(\alpha \cdot \log |V(G)|)$~\cite{CaroH11,HarutyunyanLNT18},
\cref{obs:obn_le_ods+1} implies that $\obn(G) \in O(\alpha \cdot \log |V(G)|)$.

For orientable domination number, it is known that
\begin{equation}
  \alpha(G) \le \odn(G) \le n - \mn(G)
  \label{eq:dom}
\end{equation}
for every $n$-vertex graph $G$~\cite{CaroH12,ChartrandVY96}.
The following counterpart for orientable burning number can be shown in almost the same way.
\begin{lemma}
\label{lem:alpha_le_obn_le_ec+1}
For every $n$-vertex graph $G$, $\alpha(G) \le \obn(G) \le n - \mn(G) + 1$.
\end{lemma}
\begin{proof}
The second inequality follows from the corresponding one in \cref{eq:dom} and \cref{obs:obn_le_ods+1}.

To show the first inequality, let $I$ be a maximum independent set of $G$.
Since $I$ is independent,
there is an orientation $\ori{G}$ of $G$ such that all vertices in $I$ are of in-degree~$0$.
Since every burning sequence has to contain all vertices of in-degree~$0$,
we have $\obn(G) \ge \bn(\ori{G}) \ge |I| = \alpha(G)$.
\end{proof}

\cref{eq:dom} and the equality $\alpha(G) = n - \mn(G)$ for {\KE} graphs imply that $\odn(G) = \alpha(G)$ for them~\cite{CaroH12}.
On the other hand, because of the additive factor $+1$ in \cref{lem:alpha_le_obn_le_ec+1},
we only know that $\obn(G) \in \{\alpha(G), \alpha(G)+1\}$ for {\KE} graphs.
In \cref{sec:KE-graph}, we present a polynomial-time algorithm that determines which is the case.

A \emph{tournament} is an orientation of a complete graph.
A \emph{king} of a tournament $T = (V,A)$ is a vertex $v \in V$ such that $N^{+}_{2}[v] = V$~\cite{Maurer80}.
The following fact due to Landau~\cite{Landau53} is well known.
\begin{proposition}
[\cite{Landau53}]
\label{prop:tournament-radius}
In a tournament, every vertex with the maximum out-degree is a king.
\end{proposition}

By using \cref{prop:tournament-radius}, we can show the following upper bound of orientable burning number in terms of clique cover number.
\begin{lemma}
\label{lem:obn_le_cc+2}
For every graph $G$, $\obn(G) \le \cc(G) + 2$.
\end{lemma}
\begin{proof}
Let $\mathcal{C} = \{C_{2},\dots,C_{\cc+1}\}$ be a minimum clique cover of $G = (V,E)$. (Notice the shift in the numbering.)
Given an orientation $\ori{G}$ of $G$,
we construct a sequence $\langle w_{0}, w_{1}, \dots, w_{\cc + 1}\rangle$
by setting $w_{0}$ and $w_{1}$ to arbitrary vertices
and setting $w_{i}$ with $i \ge 2$ to a king of $\ori{G}[C_{i}]$.
Since $C_{i} \subseteq N^{+}_{2}[w_{i}]$ for $2 \le i \le \cc+1$,
it holds that $\bigcup_{0 \le i \le \cc+1} N^{+}_{i}[w_{i}] = V$.
Thus, $\langle w_{0}, w_{1}, \dots, w_{\cc(G) + 1}\rangle$ is a burning sequence of $\ori{G}$ with length $\cc(G)+2$.
\end{proof}

Recall that $\cc(G) = \alpha(G)$ holds for every perfect graph $G$.
Hence, \cref{lem:alpha_le_obn_le_ec+1,lem:obn_le_cc+2} imply the following almost tight bounds for perfect graphs.
\begin{corollary}
\label{cor:perfect_graphs}
For every perfect graph $G$, $\alpha(G) \le \obn(G) \le \alpha(G) + 2$.
\end{corollary}
Since the independence number of a perfect graph $G$ can be computed in polynomial time~\cite{GrotschelLS88},
one can compute in polynomial time a value $b$ such that $b \le \obn(G) \le b +2$.
We left the complexity of \textsc{Orientable Burning Number} on perfect graphs unsettled.

\cref{cor:perfect_graphs} implies that $\obn(K_{n}) \le 3$.
As $\odn(K_{n}) \simeq \log_{2} n$, this example shows that the gap between $\obn(G)$ and $\odn(G)$ can be arbitrarily large.
It is easy to see that the lower bound $\alpha(G)$ is not always tight.
For example, $\obn(P_{2}) = 2 = \alpha(P_{2}) + 1$, where $P_{n}$ is the path on $n$ vertices.
To give examples of graphs $G$ with $\obn(G) = \alpha(G) + 2$,
we show that $\obn(K_{n}) = 3$ for $n \ge 5$.
To this end, we need the following simple observation.
\begin{observation}
\label{obs:indsub}
If $H$ is an induced subgraph of $G$, then $\obn(H) \le \obn(G)$.
\end{observation}
\begin{proof}
It suffices to show the statement for the case where $V(G) = V(H) \cup \{v\}$ with $v \notin V(H)$.
Let $\ori{H}$ be an orientation of $H$ that satisfies $\obn(H) = \bn(\ori{H})$.
Let $\ori{G}$ be an orientation of $G$
such that $A(\ori{H}) \subseteq A(\ori{G})$
and that the arcs in $A(\ori{G}) \setminus A(\ori{H})$ are oriented toward the new vertex $v$.
It suffices to show that $\bn(\ori{H}) \le \bn(\ori{G})$.
Let $\sigma = \langle w_{0}, \dots, w_{b-1} \rangle$ be a burning sequence of $\ori{G}$.
Since $v$ has out-degree~$0$ in $\ori{G}$,
we have $N^{+}_{i, \ori{H}}[w_{i}] = N^{+}_{i, \ori{G}}[w_{i}] \setminus \{v\}$
for all $i$ with $w_{i} \ne v$.
Thus, if $v \notin \{w_{0}, \dots, w_{b-1}\}$, then $\sigma$ is a burning sequence of $\ori{H}$.
Assume that $w_{j} = v$ for some $j$.
We obtain a sequence $\sigma'$ from $\sigma$ by replacing the $j$th vertex with an arbitrary vertex in $\ori{H}$.
Since $N^{+}_{j, \ori{G}}[w_{j}] = \{v\}$, $\sigma'$ is a burning sequence of $\ori{H}$.
\end{proof}

\begin{lemma}
\label{lem:complete-graphs}
$\obn(K_{1}) = 1$,
$\obn(K_{n}) = 2$ for $2 \le n \le 4$, and
$\obn(K_{n}) = 3$ for $n \ge 5$.
\end{lemma}
\begin{proof}
Clearly, $\obn(K_{1}) = 1$.
Assume first that $2 \le n \le 4$.
We can see that $\obn(K_{n}) > 1$ as a fire of radius~$0$ can burn only one vertex.
Let $\ori{K}_{n}$ be an orientation of $K_{n}$.
Observe that $\ori{K}_{n}$ has a vertex of out-degree at least $\lceil (n-1)/2 \rceil$ as there are $n(n-1)/2$ arcs.
Hence we can burn at least $1 + \lceil (n-1)/2 \rceil$ vertices by placing a fire of radius~$1$ at a vertex of maximum out-degree.
This is already enough for $n = 2$. If $n \in \{3,4\}$, there can be one unburnt vertex, which can be burned by a fire of radius~$0$.

Next assume that $n \ge 5$.
By \cref{cor:perfect_graphs,obs:indsub}, it suffices to show that $\obn(K_{5}) > 2$.
Let $V(K_{5}) = \{0,1,2,3,4\}$.
Let $\ori{K}_{5}$ be the orientation of $K_{5}$ that has $A(\ori{K}_{5}) = \{(i, i+1), (i, i+2) \mid 0 \le i \le 4\}$, where the addition is modulo $5$.
See \cref{fig:K5}.
This orientation shows that $\obn(K_{5}) > 2$ since a fire of radius~$1$ burns only three vertices.
\begin{figure}[htb]
\centering
\includegraphics{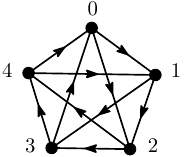}
\caption{An optimal orientation of $K_{5}$.}
\label{fig:K5}
\end{figure}
\end{proof}


\subsection{Polynomial-time algorithm for {\KE} graphs}
\label{sec:KE-graph}

We now present a polynomial-time algorithm for {\KE} graphs.
Recall that $\obn(G) \in \{\alpha(G), \alpha(G)+1\}$ for every {\KE} graph $G$.
Intuitively, we show that the fire of radius~$0$ (often called $w_{0}$ in our exposition) is useful in most of the cases
and the case $\obn(G) = \alpha(G)+1$ rarely happens.

\begin{theorem}
\label{thm:ke-graphs}
Let $G$ be a {\KE} graph with more than four vertices, then
\[
  B(G) =
  \begin{cases}
    \alpha(G)+1 & \text{if $G = m P_{2}$}, \\
    \alpha(G)   & \text{otherwise},
  \end{cases}
\]
where $m = |E|$ and $m P_{2}$ is the disjoint union of $m$ edges.
\end{theorem}

This immediately gives the complexity of \textsc{Orientable Burning Number} for {\KE} graphs.
\begin{corollary}
\label{cor:ke-graphs}
\textsc{Orientable Burning Number} on {\KE} graphs is solvable in polynomial time.
\end{corollary}

\begin{proof}[\cref{thm:ke-graphs}]
Let $G = (V,E)$ be an $n$-vertex {\KE} graph.
Recall that $\alpha(G) = n - \mn(G)$ and $\obn(G) \in \{\alpha(G), \alpha(G)+1\}$.

First assume that $G = m P_{2}$. In this case, $\alpha(G) = m$ holds.
Suppose to the contrary that $\obn(G) = \alpha(G)$.
To burn the whole graph $G$, each connected component has to contain a fire.
Since we have $\alpha(G)$ connected components and $\alpha(G)$ fires,
each connected component contains exactly one fire.
However, since each connected component contains two vertices,
the one with the fire of radius~$0$ is not completely burned.
Thus, $\obn(G) \ge \alpha(G) + 1$.

Next assume that $G \ne m P_{2}$ and $G$ has no perfect matching.
Let $M = \{e_{1}, \dots, e_{\mn(G)}\}$ be a maximum matching of $G$.
By the non-existence of a perfect matching, $2\mn(G) < n$ holds,
and thus there exist $n - 2\mn(G) = \alpha(G) - \mn(G) > 0$ vertices not covered by $M$.
Let $\ori{G}$ be an orientation of $G$.
We set $w_{0},\dots,w_{\alpha(G) - \mn(G) -1}$ to the $\alpha(G)-\mn(G)$ vertices not covered by $M$,
and set $w_{\alpha(G) - \mn},\dots,w_{\alpha(G) -1}$ to the tails\footnote{%
The tail of an arc $a = (u,v)$ is the vertex $u$, which has $a$ as an out-going arc.}
of the arcs corresponding to $e_{1}, \dots, e_{\mn(G)}$.
The constructed sequence $\langle w_{0}, \dots, w_{\alpha(G)-1} \rangle$ is a burning sequence of $\ori{G}$ with length $\alpha(G)$.

Finally, we consider the case where $G \ne m P_{2}$ and $G$ has a perfect matching $M$.
Observe that in this case $\mn(G) = |V|/2$, and thus $\alpha(G) = n - \mn(G) = |V|/2$.
We set $M = \{e_{0}, \dots, e_{\alpha(G)-1}\}$.
Note that $|M| \ge 3$ as $|V| > 4$.
Since $G \ne m P_{2}$, $G$ has another edge $f \notin M$.
Without loss of generality, we assume that $f$ has one endpoint in $e_{0}$ and the other in $e_{2}$.
That is, $e_{0}$, $f$, and $e_{2}$ together form $P_{4}$.
Observe that every orientation of $P_{4}$ can be burned by fires of radii $0$ and $2$ (see \cref{fig:ke-graph}).
For an orientation $\ori{G}$ of $G$,
we construct a sequence $\sigma = \langle w_{0}, w_{1}, \dots, w_{\alpha(G) - 1}\rangle$
by setting $w_{0}$ and $w_{2}$ appropriately in the $P_{4}$ formed by $e_{0}$, $f$, $e_{2}$,
and setting $w_{i}$ to the tail of the arc corresponding to $e_{i}$ for each $i \in \{0, \dots, \alpha(G)-1\} \setminus \{0,2\}$.
The sequence $\sigma$ is a burning sequence of $\ori{G}$ with length $\alpha(G)$.
\begin{figure}[htb]
\centering
\includegraphics{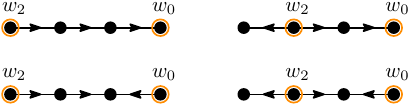}
\caption{All orientations of $P_{4}$ can be burned by fires of radii $0$ and $2$.
(By symmetry, the direction of the central edge is fixed.)}
\label{fig:ke-graph}
\end{figure}
\end{proof}

\section{Hardness of the problem}
\label{sec:hardness}
Now we demonstrate that \textsc{Orientable Burning Number} is intractable in general.
We present two reductions implying that
\begin{itemize}
 \item OBN is NP-hard, and
 \item OBN is W[1]-hard parameterized by the target burning number $b$.
\end{itemize}
The reductions are similar and they follow the same idea:
reduce the problem to \textsc{Independent Set} by adding a sufficiently large number of isolated vertices.

We can see that our reduction showing the W[1]-hardness also shows NP-hardness in general.
However, we present the separate reduction for NP-hardness as it has a wider range of applications.
Basically, our reduction for W[1]-hardness  works only for dense graphs,
while the one for NP-hardness works also for sparse graphs like planar graphs.

Technically, the reduction for W[1]-hardness is a little more involved as it has to control the number of additional isolated vertices
to keep the target burning number small.
Thus, we first prove W[1]-hardness and then show a similar proof for NP-hardness.

\subsection{W[1]-hardness parameterized by $b$}
We reduce the following problem, which is known to be W[1]-complete parameterized by the solution size $k$~\cite{CyganFKLMPPS15},
to OBN parameterized by the target burning number $b$.
\begin{description}
  \item[Problem:] \textsc{Multicolored Independent Set} (MCIS)
  \item[Input:] An undirected graph $G = (V, E)$ and a partition $(V_{1},\dots,V_{k})$ of $V$ into cliques.
  \item[Question:] Does $G$ contain an independent set of size $k$?
\end{description}

\begin{theorem}
\label{thm:W[1]-hardness}
\textsc{Orientable Burning Number} on connected graphs is $\mathrm{W[1]}$-hard parameterized by the target burning number $b$.
\end{theorem}
\begin{proof}
Let $(G,V_{1},\dots,V_{k})$ be an instance of MCIS\@.
Let $H$ be the connected graph obtained from $G$
by first adding a set $I$ of four isolated vertices and then adding a universal vertex $u$.\footnote{%
If we do not need the connectivity, we can omit the universal vertex and the proof becomes slightly simpler.}
We prove that $(H, k+4)$ is a yes-instance of OBN if and only if $(G,k)$ is a yes-instance of MCIS\@.

To show the if direction, assume that $G$ has an independent set $S$ of size~$k$.
Since $S \cup I$ is an independent set of $H$, there is an orientation $\ori{H}$ of $H$ such that each vertex in $S \cup I$ has in-degree~$0$,
and thus $\bn(\ori{H}) \ge k+4$.

In the following, we show the only-if direction.
Assume that $(H, k+4)$ is a yes-instance of OBN and an orientation $\ori{H}$ of $H$ satisfies $\bn(\ori{H}) \ge k+4$.

We construct a sequence $\sigma = \langle w_{0}, \dots, w_{k+3} \rangle$ as follows.
If all vertices in $I$ are of in-degree~$0$, then we set $w_{0},w_{1},w_{2},w_{3}$ to the vertices in $I$.
Otherwise, we set $w_{3}$ to $u$ and set $w_{0}, w_{1}, w_{2}$ to three vertices of $I$ including the ones of in-degree~$0$ (if any exist).
For $1 \le i \le k$, we set $w_{i+3}$ to a king of the tournament $\ori{H}[V_{i}]$.
Recall that a king of a tournament can reach the other vertices in the tournament in at most two steps.
Recall also that every tournament has a king, which can be found in polynomial time (see \cref{prop:tournament-radius}).
We can see that $\sigma$ is a burning sequence of $\ori{H}$ (with length $k+4$) as follows.
Each vertex of in-degree~$0$ in $\{u\} \cup I$, if any exists, is burned by itself,
and the other vertices in $\{u\} \cup I$ are burned by $w_{3}$.
For $1 \le i \le k$, $w_{i+3}$ burns $V_{i}$ as $i + 3 > 2$ and $w_{i+3}$ is a king of $\ori{H}[V_{i}]$.
Since $\bn(\ori{H}) \ge k+4$, $\sigma$ is a shortest burning sequence of $\ori{H}$.

Now we show that $\{w_{4}, \dots, w_{k+3}\}$ is an independent set of $G$,
which implies that $(G,V_{1},\dots,V_{k})$ is a yes-instance of MCIS\@.
Suppose to the contrary that $G$ has an edge between vertices $w_{p}, w_{q} \in \{w_{4}, \dots, w_{k+3}\}$.
By symmetry, we may assume that $(w_{p}, w_{q}) \in A(\ori{H})$.
Let $\sigma' = \langle w'_{0}, \dots, w'_{k+2} \rangle$ be the sequence obtained from $\sigma$ by skipping $w_{q}$, that is, a sequence $\sigma'$ defined as
\[
  w'_{i} =
  \begin{cases}
    w_{i}   & 0 \le i \le q-1, \\
    w_{i+1} & q \le i \le k+2.
  \end{cases}
\]
We show that $\sigma'$ is a burning sequence of $\ori{H}$, which contradicts that $\sigma$ is shortest.

As $\{w'_{0}, w'_{1}, w'_{2}, w'_{3}\} = \{w_{0}, w_{1}, w_{2}, w_{3}\}$,
we have $\bigcup_{0 \le i \le 3} N^{+}_{i}[w'_{i}] \supseteq \{u\} \cup I$.
For $4 \le i \le q-1$, we have $w'_{i} = w_{i}$, and thus $N^{+}_{i}[w'_{i}] = N^{+}_{i}[w_{i}] \supseteq V_{i-3}$.
For $q \le i \le k+2$, $w'_{i} = w_{i+1}$ is a king of $V_{i-2}$,
and thus $N^{+}_{i}[w'_{i}] \supseteq V_{i-2}$.
The discussion so far implies that
\[
  V(G) \setminus V_{q-3} \subseteq \bigcup_{0 \le i \le k+2} N^{+}_{i}[w'_{i}].
\]
We now show that $V_{q-3}$ is also burned by $\sigma'$.
\begin{myclaim}
\label{clm:V_q-3}
$V_{q-3} \subseteq N^{+}_{p-1}[w'_{p-1}] \cup N^{+}_{p}[w'_{p}]$.
\end{myclaim}
\begin{subproof}[\cref{clm:V_q-3}]
Since $w_{q}$ is a king of $V_{q-3}$, we have $V_{q-3} \subseteq N_{2}^{+}[w_{q}]$.
As $(w_{p}, w_{q}) \in A(\ori{H})$, it holds that $N_{2}^{+}[w_{q}] \subseteq N_{3}^{+}[w_{p}]$.
Since $p \ge 4$, $N_{3}^{+}[w_{p}] \subseteq N_{p-1}^{+}[w_{p}]$ holds.
The chain of inclusions implies that $V_{q-3} \subseteq N_{p-1}^{+}[w_{p}]$.
Since $w_{p} \in \{w'_{p-1}, w'_{p}\}$,
we have $V_{q-3} \subseteq N_{p-1}^{+}[w'_{p-1}]$ or $V_{q-3} \subseteq N_{p-1}^{+}[w'_{p}] \subseteq N_{p}^{+}[w'_{p}]$,
as desired.
\end{subproof}
By \cref{clm:V_q-3}, we conclude that $\sigma'$ is a burning sequence of $\ori{H}$ with length $k+3$.
This contradicts the assumption that $\sigma$ is a shortest one.
\end{proof}

\subsection{NP-hardness}

We reduce the following NP-complete problem to OBN\@.
\begin{description}
  \item[Problem:] \textsc{Independent Set}
  \item[Input:] An undirected graph $G = (V, E)$ and an integer $k$.
  \item[Question:] Does $G$ contain an independent set of size $k$?
\end{description}

\textsc{Independent Set} is quite well studied
and known to be NP-complete on many restricted graph classes
such as planar graphs of maximum degree~$3$~\cite{GareyJ77rst}.
\cref{thm:NP-hardness} below shows that the hardness of \textsc{Independent Set}
on a graph class can be translated to the hardness of OBN on the same graph class,
under an assumption that the class does not change by additions of isolated vertices.

The proof of \cref{thm:NP-hardness} is quite similar to that of \cref{thm:W[1]-hardness}.
The main difference is that here we can increase the target burning number as much as we like.
This makes the discussion very simple; e.g., we do not need the concept of kings any more,
and thus the hardness can be proved for almost all graph classes for which \textsc{Independent Set} is hard.
\begin{theorem}
\label{thm:NP-hardness}
Let $\mathcal{G}$ be a graph class such that \textsc{Independent Set} is $\mathrm{NP}$-complete on $\mathcal{G}$.
If $\mathcal{G}$ is closed under additions of isolated vertices,
then \textsc{Orientable Burning Number} on $\mathcal{G}$ is $\mathrm{NP}$-hard.
\end{theorem}
\begin{proof}
Let $(G,k)$ be an instance of \textsc{Independent Set}, where $G \in \mathcal{G}$ and $|V(G)| = n$.
Let $H$ be the graph obtained from $G$ by adding a set $I$ of $n$ isolated vertices.
Since $\mathcal{G}$ is closed under additions of isolated vertices, $H \in \mathcal{G}$ holds.
To prove the theorem,
it suffices to show that $(H, k+n)$ is a yes-instance of OBN if and only if $(G,k)$ is a yes-instance of \textsc{Independent Set}.

If $G$ has an independent set $S$ of size $k$,
then an orientation of $H$ that makes all vertices in the independent set $S \cup I$ in-degree $0$ shows that $\bn(H) \ge k + n$.

Conversely, assume that $(H, k+n)$ is a yes-instance of OBN
and there is an orientation $\ori{H}$ of $H$ such that $\bn(\ori{H}) \ge k+n$.
Let $\sigma = \langle w_{0}, \dots, w_{b-1} \rangle$ be a shortest burning sequence of $\ori{H}$, where $b \ge k+n$.
Without loss of generality, we may assume that $\{w_{0}, w_{1}, \dots, w_{n-1}\} = I$
and $\{w_{n}, w_{n+1}, \dots, w_{b-1}\} \subseteq V(G)$.
Let $S = \{w_{n}, w_{n+1}, \dots, w_{b-1}\}$.
It suffices to show that $S$ is an independent set of $G$ as $|S| = b - n \ge k$.
Suppose to the contrary that $G$ has an edge between vertices $w_{p}, w_{q} \in S$.
By symmetry, we may assume that $(w_{p}, w_{q}) \in A(\ori{H})$.
Let $\sigma' = \langle w'_{0}, \dots, w'_{b-2} \rangle$ be the sequence obtained from $\sigma$ by skipping $w_{q}$.
That is, $\sigma'$ is defined as follows:
\[
  w'_{i} =
  \begin{cases}
    w_{i}   & 0 \le i \le q-1, \\
    w_{i+1} & q \le i \le b-2.
  \end{cases}
\]
We show that $\sigma'$ is a burning sequence of $\ori{H}$, which contradicts that $\sigma$ is a shortest one.
For $0 \le i \le q-1$, we have $w'_{i} = w_{i}$, and thus $N^{+}_{i}[w'_{i}] = N^{+}_{i}[w_{i}]$.

Let us consider the case $q \le i \le b-2$, where $w'_{i} = w_{i+1}$.
Since $i \ge q \ge n$ and a longest directed path in $\ori{H}$ has length less than $n$,
we have $N^{+}_{i}[w'_{i}] = N^{+}_{i}[w_{i+1}] = N^{+}_{i+1}[w_{i+1}]$.

Furthermore, since $p, q \ge n$ and $(w_{p}, w_{q}) \in A(\ori{H})$,
it holds that $N^{+}_{q}[w_{q}] \subseteq N^{+}_{q+1}[w_{p}] = N^{+}_{p}[w_{p}]$.
Hence, we can conclude that
\begin{align*}
  \bigcup_{0 \le i \le b-1} N^{+}_{i}[w_{i}]
  &=
  \left(\bigcup_{0 \le i \le q-1} N^{+}_{i}[w_{i}]\right)
  \cup
  N^{+}_{q}[w_{q}]
  \cup
  \left(\bigcup_{q+1 \le i \le b-1} N^{+}_{i}[w_{i}]\right)
  \\
  &=
  \left(\bigcup_{0 \le i \le q-1} N^{+}_{i}[w_{i}]\right)
  \cup
  \left(\bigcup_{q+1 \le i \le b-1} N^{+}_{i}[w_{i}]\right)
  \\
  &\subseteq
  \bigcup_{0 \le i \le b-2} N^{+}_{i}[w'_{i}].
\end{align*}
This implies that $\sigma'$ is a burning sequence of $\ori{H}$.
\end{proof}

As an application of \cref{thm:NP-hardness}, we can show the following corollary.
(Recall that \textsc{Independent Set} is NP-complete on planar graphs of maximum degree~$3$~\cite{GareyJ77rst}.)
\begin{corollary}
\label{cor:NP-hard_cubic-planar}
\textsc{Orientable Burning Number} is NP-hard on planar graphs of maximum degree~$3$.
\end{corollary}


\section{Structural parameterizations}
\label{sec:parameterizations}

In this section, we consider some structural parameterizations of \textsc{Orientable Burning Number}.
Given \cref{thm:W[1]-hardness}, which shows that OBN is intractable when parameterized by the target burning number $b$,
it is natural to consider the problem parameterized by some structural parameters of the input graph.

\subsection{Parameterizations combined with $b$}
The first observation is that some sparseness parameters combined with $b$ make the problem tractable.
\begin{observation}
\label{obs:sparse}
Let $\mathcal{G}$ be a class of graphs with a constant $c_{\mathcal{G}} > 1$ such that
every graph $G \in \mathcal{G}$ satisfies $\alpha(G) \ge |V(G)| / c_{\mathcal{G}}$.
When parameterized by $c_{\mathcal{G}}$ plus the target burning number $b$,
\textsc{Orientable Burning Number} on $\mathcal{G}$ is fixed-parameter tractable.
\end{observation}
\begin{proof}
Let $G$ be an $n$-vertex graph in $\mathcal{G}$.
If $b \le n/c_{\mathcal{G}}$, then $\alpha(G) \ge b$ and thus $(G,b)$ is a yes-instance of OBN\@.
If $b > n/c_{\mathcal{G}}$, then $n < b \cdot c_{\mathcal{G}}$, and thus $(G,b)$ itself is a kernel with less than $b \cdot c_{\mathcal{G}}$ vertices.
\end{proof}

It is known that $\alpha(G) \ge n/(d+1)$
for every $n$-vertex graph $G$ with average degree $d$~\cite{Caro1979new,Wei1981}.
Thus \cref{obs:sparse} implies the following corollary.
\begin{corollary}
\label{cor:fpt_avgdeg+b}
\textsc{Orientable Burning Number} is fixed-parameter tractable parameterized by $b$ plus average degree.
\end{corollary}

\cref{cor:fpt_avgdeg+b} implies that OBN is fixed-parameter tractable parameterized by $b + \text{treewidth}$,
and OBN on planar graphs is fixed-parameter tractable parameterized by $b$.
Recall that OBN is NP-hard on planar graphs even if the maximum degree is~$3$ (\cref{cor:NP-hard_cubic-planar}).
On the other hand, the parameterized complexity of OBN parameterized solely by treewidth remains unsettled.

\subsection{Parameterizations without $b$}

Now we consider structural parameterizations \emph{not} combined with $b$.
As the first step in this direction, we consider parameters less general than treewidth such as vertex cover number.
In some sense, vertex cover number is one of the most restricted parameters that is always larger than or equal to treewidth (see, e.g., \cite{GimaHKKO22}).

We show that OBN parameterized solely by vertex cover number is fixed-parameter tractable.
Our proof is actually for a slightly more general case,
where the parameter is cluster vertex deletion number plus clique number.
In the rest of this section, we prove the following theorem.
\begin{theorem}
\label{thm:cvd+omega}
\textsc{Orientable Burning Number} is fixed-parameter tractable parameterized by cluster vertex deletion number plus clique number.
\end{theorem}

\cref{thm:cvd+omega,obs:paras_ineq} imply the fixed-parameter tractability parameterized by vertex cover number.
\begin{corollary}
\label{cor:vc}
\textsc{Orientable Burning Number} is fixed-parameter tractable parameterized by vertex cover number.
\end{corollary}

\subsubsection{Proof of \cref{thm:cvd+omega}.}
In the proof, we use the theory of monadic second-order logic on graphs (MSO$_{2}$), which will be introduced right before we use it.
If we allow an MSO$_{2}$ formula to have length depending on $b$, it is not difficult to express OBN\@.
However, this only implies the fixed-parameter tractability of OBN parameterized by a parameter combined with $b$.
To avoid the dependency on $b$, we have to bound the length of an MSO$_{2}$ formula with a function not depending on $b$.
To this end, we make a series of observations to bound the number of \emph{parts} not used in a \emph{good} burning sequence,
then represent the problem by expressing the unused parts instead of the used parts.

\paragraph{Useful observations.}
In the following, let $(G, b)$ be an instance of OBN\@.
Let $\omega$ be the clique number of $G$; that is, $\omega = \omega(G)$.
Let $S$ be a cluster vertex deletion set of $G$ with size $s = \cvd(G)$.
Our parameter is $k \coloneqq \omega + s$.
Note that finding $S$ is fixed-parameter tractable parameterized by $s$~\cite{HuffnerKMN10}, and thus by $k$ as well.
Let $C_{1}, \dots, C_{p}$ be the connected components of $G - S$, which are complete graphs.
When we are dealing with an orientation $\ori{G}$ of $G$, we sometimes mean by $C_{i}$ the tournament $\ori{G}[V(C_{i})]$.
For example, we may say ``a king of $C_{i}$.''

\begin{myclaim}
\label{clm:small-b}
If $b \le p$, then $(G,b)$ is a yes-instance.
\end{myclaim}
\begin{proof}
By picking arbitrary one vertex from each $C_{i}$, we can construct an independent set of size $p$.
By \cref{lem:alpha_le_obn_le_ec+1}, $\obn(G) \ge \alpha(G) \ge p \ge b$.
\end{proof}

\begin{myclaim}
\label{clm:large-b}
If $b > p + s + 2$, then $(G,b)$ is a no-instance.
\end{myclaim}
\begin{proof}
Let $\ori{G}$ be an orientation of $G$.
It suffices to show that $\bn(\ori{G}) \le p+s+2$.
For each $C_{i}$, we place a fire of radius at least~$2$ at a king of $C_{i}$.
For each vertex in $S$, we place a fire of arbitrary radius.
If we have not used the fires of radii~$0$ and $1$,
then we place them at arbitrary vertices.
\end{proof}

By \cref{clm:small-b,clm:large-b}, we may assume that
\[
  p < b \le p + s +2.
\]

Let $\ell$ be the length of a longest path in $G$.
We assume that $\ell \ge 1$ since otherwise $G$ cannot have any edge and the problem becomes trivial.
Note that in every orientation $\ori{G}$ or $G$,
the length of a longest directed path is at most $\ell$.
\begin{myclaim}
\label{clm:longest-path}
$\ell \le s\omega + s + \omega -1$.
\end{myclaim}
\begin{proof}
Let $P$ be a longest path in $G$.
Since $P$ can visit at most $|S| + 1$ connected components of $G-S$,
we have $|V(P)| \le |S| + (|S|+1) \omega$ as each $C_{i}$ is a complete graph.
The claim follows as $|S| = s$ and $|E(P)| = |V(P)|-1$.
\end{proof}

In a burning sequence of an orientation of $G$,
we call a fire of radius at least $\ell$ a \emph{large fire}.
Note that a large fire $w$ burns all vertices that can be reached from $w$ in the orientation
as no directed path in the orientation is longer than~$\ell$.

In the following, we focus on burning sequences of length $b-1$
since we are going to express the \emph{non-existence} of such sequences.
Let $L = \max\{0, b-1 - \ell\}$;
that is $L$ is the number of large fires in a sequence of length $b-1$.
Observe that $L \le p+s$ as $b-1-\ell \le b-2 \le p+s$.

A burning sequence of an orientation of $G$ is \emph{good}
if the following conditions are satisfied:
\begin{enumerate}
  \item two large fires do not have the same position;
  \item each $C_{j}$ contains at most one large fire;
  \item if a large fire is placed in some $C_{h}$, then it is placed at a king of $C_{h}$.
\end{enumerate}

\begin{myclaim}
\label{clm:good}
Let $\ori{G}$ be an orientation of $G$.
If $\ori{G}$ admits a burning sequence with length $b-1$,
then there is a good burning sequence of $\ori{G}$ with the same length.
\end{myclaim}
\begin{proof}
From a burning sequence $\sigma$ of $\ori{G}$ with length $b-1$,
we first construct a sequence $\sigma_{1}$ that satisfies the first condition of the goodness.
We repeatedly find two large fires placed at the same vertex
and then replace arbitrary one of them with another vertex not occupied by any large fire.
The replacement is possible as $L$ is not larger than the number of vertices.
Since two large fires placed at the same vertex burn the same set of vertices,
the obtained sequence is still a burning sequence of $\ori{G}$.
When there is no pair of large fires occupying the same vertex, we stop this phase and call the resultant sequence~$\sigma_{1}$.

Next we modify $\sigma_{1}$ to obtain a sequence $\sigma_{2}$ that satisfies the first and second conditions.
Assume that two large fires $w_{i}$ and $w_{j}$ are placed in the same connected component $C_{h}$ of $G-S$
and that $(w_{i}, w_{j}) \in A(\ori{G})$. (Recall that $C_{h}$ is a complete graph.)
Since $w_{i}$ is a large fire, it burns every vertex that can be reached from $w_{i}$.
In particular, $w_{i}$ burns all vertices reachable from $w_{j}$.
Hence, $w_{j}$ is useless for burning the graph.
We replace $w_{j}$ with another vertex $v$ such that
$v$ is not occupied by any large fire
and
if $v$ belongs to some $C_{h'}$, then there is no large fire belonging to $C_{h'}$ prior to the replacement.
This is always possible as $L \le p+s$.
We exhaustively apply this replacement procedure and get $\sigma_{2}$,
which satisfies the first and second conditions of the goodness.

Finally, we obtain a sequence $\sigma_{3}$ from $\sigma_{2}$
by replacing each large fire that is placed in some $C_{h}$ with a king of $C_{h}$.
We can see that $\sigma_{3}$ is a burning sequence of $\ori{G}$
since the new large fire placed at a king of $C_{h}$ burns all vertices reachable form the king and the king can reach all vertices in $C_{h}$.
Since $\sigma_{3}$ satisfies all conditions of goodness and has the same length as $\sigma$, the claim holds.
\end{proof}

If $\sigma$ is a good burning sequence of an orientation $\ori{G}$ of $G$ with length $b-1$,
then the sum of the following two numbers is $p+s-L$:
\begin{itemize}
  \item the number of vertices in $S$ not occupied by the large fires of $\sigma$, and
  \item the number of connected components of $G-S$ not including large fires of~$\sigma$.
\end{itemize}
Since $L = \max\{0, b-1-\ell\}$ and $p < b$, it holds that $p+s-L < s+\ell+1$.
Since $\ell$ can be bounded from above by a function of $k = s + \omega$, so is $p+s-L$.
Thus our MSO$_{2}$ formula can have length depending on $p+s-L$, the number of \emph{unused parts}.

\paragraph{MSO$_{2}$ expressions.}
We now express the problem in the monadic second-order logic on graphs.
A formula in the monadic second-order logic on graphs, denoted MSO$_{2}$,
can have variables representing vertices, vertex sets, edges, and edge sets.
As atomic formulas, we can use
the equality $x = y$, the inclusion $x \in X$,
the adjacency relation $\mathtt{adj}(x, y)$ meaning that vertices $x$ and $y$ are adjacent, and
the incidence relation $\mathtt{inc}(e, x)$ meaning that a vertex $x$ is an endpoint of an edge $e$.
Atomic formulas can be recursively combined by the usual Boolean connectives $\lnot$, $\land$, $\lor$, $\Rightarrow$, and $\Leftarrow$ to form an MSO$_{2}$ formula.
Furthermore, variables in an MSO$_{2}$ formula can be quantified by $\exists$ and $\forall$.
If an MSO$_{2}$ formula $\phi(X)$ with one free (vertex-set or edge-set) variable $X$ is evaluated to true
on a graph $G$ and a subset $S$ of $V(G)$ or $E(G)$, we write $G \models \phi(S)$.
It is known that, given an MSO$_{2}$ formula $\phi(X)$, a graph $G$, and a subset $S$ of $V(G)$ or $E(G)$,
the problem of deciding whether $G \models \phi(S)$ is fixed-parameter tractable parameterized
by the length of $\phi$ plus the treewidth of $G$~\cite{ArnborgLS91,BoriePT92,Courcelle90mso1}.

In the following, we construct an MSO$_{2}$ formula $\phi(X)$ such that $G \models \phi(S)$
if and only if $(G,b)$ is a yes-instance of OBN, $S$ is a minimum cluster vertex deletion set,
and the following assumptions, which we justified above, are satisfied:
\begin{itemize}
  \item $p < b \le p + s +2$;
  \item the sum of the numbers of unused vertices in $S$ and unused connected components of $G-S$ is $p+s-L$,
\end{itemize}
where
$s = \cvd(G) = |S|$, $\omega = \omega(G)$,
$p$ is the number of connected components in $G - S$, and
$L$ is the number of large fires in a sequence of length $b-1$.
Also, we show that the length of $\phi(X)$ is bounded from above by a function of $k = s + \omega$.
Since the treewidth of $G$ is at most $s + \omega$ (see \cite{GimaHKKO22}), this implies \cref{thm:cvd+omega}.

The formula $\phi(X)$ asks whether there exists an orientation $\ori{G}$ of $G$ such that
no sequence $\langle w_{0}, \dots, w_{b-2} \rangle$ of length $b-1$ is a good burning sequence of~$\ori{G}$.
In MSO$_{2}$, we can handle orientations of $k$-colorable graphs  by first assuming a \emph{default} orientation using a $k$-coloring and
then represent the \emph{reversed} edges by a set of edges~\cite{Courcelle95,JaffkeBHT17}.\footnote{%
There is another way for handling orientation by using a variant of MSO$_{2}$ defined for directed graphs,
where we can fix an arbitrary orientation first (without using a $k$-coloring) and then represent reversed edges by an edge set. See e.g., \cite{EggemannN12}.}
More precisely, such a formula first expresses a $k$-coloring and a set of reversed edge,
and then it considers each edge as oriented from the vertex with a smaller label to the one with a larger label
if and only if the edge is not a reversed one.
Note that \cref{obs:paras_ineq} implies that $G$ is $k$-colorable.
We use this technique and thus $\phi(X)$ has the following form:
\[
  \phi(X) = \exists V_{1}, \dots, V_{k} \subseteq V, \, \exists F \subseteq E \colon
  \texttt{proper-coloring}(V_{1}, \dots, V_{k}) \land \phi_{1},
\]
where $\texttt{proper-coloring}(V_{1}, \dots, V_{k})$ expresses that $V_{1}, \dots, V_{k}$ is a proper $k$-coloring of $G$.
The formula $\texttt{proper-coloring}(V_{1}, \dots, V_{k})$ can be defined as
\[
  \texttt{proper-coloring}(V_{1}, \dots, V_{k})
  =
  \mathtt{part}(V_{1}, \dots, V_{k})
  \land
  \bigvee_{1 \le i \le k}
  \mathtt{ind}(V_{i}),
\]
where $\mathtt{part}(V_{1}, \dots, V_{k})$ and $\mathtt{ind}(V_{i})$ are defined as follows
\begin{align*}
  \mathtt{part}(V_{1}, \dots, V_{k})
  &= \forall v \in V \colon
  \bigvee_{1 \le i \le k} (v \in V_{i})
  \land
  \bigwedge_{1 \le i < j  \le k} \lnot (v \in V_{i} \land v \in V_{j}),
  \\
  \mathtt{ind}(V_{i})
  &=
  \forall u, v \in V_{i}\colon \lnot \mathtt{adj}(u,v).
\end{align*}

The subformula~$\phi_{1}$ can use the formula $\mathtt{arc}(u,v)$ expressing that
there is an arc from $u$ to $v$ in the orientation defined by $V_{1}, \dots, V_{k}$ and $F$,
which can be defined as follows:
\begin{align*}
  \mathtt{arc}(u,v)
  ={}& \mathtt{adj}(u,v) \land
   \big(((u < v) \land \lnot \mathtt{rev}(u,v))
    \lor (\lnot (u < v) \land  \mathtt{rev}(u,v))\big),
\end{align*}
where $(u < v)$ and $\mathtt{rev}(u,v)$ are defined as
\begin{align*}
  (u < v) &= \bigvee_{1 \le i < j \le k} (u \in V_{i} \land v \in V_{j}),\\
  \texttt{rev}(u,v) &= \exists e \in F \colon \mathtt{inc}(e,u) \land \mathtt{inc}(e,v).
\end{align*}

Given an orientation defined above, the subformula $\phi_{1}$ expresses that
there is no good burning sequence of length $b-1$ for this orientation.
Indeed, we set $\phi_{1} = \lnot \phi_{2}$ and give a definition of $\phi_{2}$ that expresses there is a good burning sequence of length $b-1$.
We assume that $b-1 > \ell$ since the other case can be easily obtained from the expression of this case.
The subformula $\phi_{2}$ has the following form
\begin{align*}
  \phi_{2} &=
  \exists w_{0}, \dots, w_{\ell-1} \in V, \,
  \exists u_{1}, \dots, u_{p+s-L} \in V \colon \ \phi_{3} \land {} \\
  &  \bigwedge_{1 \le i < j \le p+s-L} (u_{i} \ne u_{j})
   \land \bigwedge_{1 \le i < j \le p+s-L} ((u_{i} \notin X) \land (u_{j} \notin X) \Rightarrow \lnot \mathtt{adj}(u_{i}, u_{j}))
\end{align*}
where $w_{0}, \dots, w_{\ell-1}$ simply correspond to the first $\ell$ fires in a (good) burning sequence
and $u_{1}, \dots, u_{p+s-L}$ correspond to the representatives of unused parts.
More precisely, if $u_{i} \in X$, then it means that $u_{i}$ is not used by any large fire;
if $u_{i} \notin X$ and thus $u_{i}$ belongs to some connected component $C$ of $G-X$,
then it means that no vertex in $C$ is used by large fires.
Note that the second line of the formula forces that $u_{1}, \dots, u_{p+s-L}$ are distinct and not chosen multiple times from one connected component of $G - X$.
(Recall that $X$ is promised to be a cluster vertex deletion set.)
Now $\phi_{3}$ expresses that every vertex is burned. Hence, it can be expressed as follows
\[
  \phi_{3} = \forall v \in V \colon \mathtt{burned}(v),
\]
where the definition of $\mathtt{burned}(v)$ is given below.

To define $\mathtt{burned}(v)$, observe that $v$ is burned if and only if one of the following conditions is satisfied:
\begin{enumerate}
  \item some $w_{i}$ ($0 \le i \le \ell-1$) has a directed path of length at most $i$ to $v$;
  \item some large fire has a directed path to $v$.
\end{enumerate}
We express the first case as $\texttt{burned-small}(v)$ and the second as $\texttt{burned-large}(v)$,
and thus $\mathtt{burned}(v) = \texttt{burned-small}(v) \lor \texttt{burned-large}(v)$.
The first case is easy to state as
\[
  \texttt{burned-small}(v) = \bigvee_{0\le i \le \ell-1} \mathtt{reachable}_{i}(w_{i}, v),
\]
where $\mathtt{reachable}_{d}(x, y)$ means that there is a directed path of length at most $d$ from $x$ to $y$,
which can be defined as
\begin{align*}
  \mathtt{reachable}_{d}(x, y) =
  \exists z_{0}, \dots, z_{d} \in V & \colon (z_{0} = x) \land (z_{d} = y) \land{} \\
  & \bigwedge_{0 \le j \le d-1} ((z_{j} = z_{j+1}) \lor \mathtt{arc}(z_{j}, z_{j+1})).
\end{align*}
On the other hand, the second case is a bit tricky as the large fires are not explicitly handled.
Recall that the vertices $u_{1}, \dots, u_{p+s-L}$ tell us which vertices in $X$ are not large fires
and which connected components of $G-X$ include no large fires.
From this information, we can determine whether a vertex $x$ is used as a large fire by setting $\texttt{large-fire}(x) = \lnot \mathtt{unused}(x)$,
where $\mathtt{unused}(x)$ is defined as
\begin{align*}
  \mathtt{unused}(x) ={}
  &\bigvee_{1 \le i \le p+s-L} (x = u_{i}) \\
  &{}\lor
  \left((x \notin X) \land \bigvee_{1 \le i \le p+s-L} ((u_{i} \notin X) \land \mathtt{adj}(x, u_{i}))\right).
\end{align*}
Note that the correctness of the second line depends on the assumption that each connected component of $G-X$ is a complete graph.
Now $\texttt{burned-large}(v)$ can be expressed as follows.
\begin{align*}
  \texttt{burned-large}(v) ={} &
  \exists x \in V \colon \texttt{large-fire}(x) \land  \mathtt{reachable}_{\ell}(x, v).
\end{align*}

The length of the entire formula $\phi(X)$ depends only on $k$, $\ell$, and $p+s-L$,
where $\ell$ and $p+s-L$ can be bounded from above by function depending only on $k$.
Therefore, the length of $\phi(X)$ depends only on $k$.
This completes the proof of \cref{thm:cvd+omega}.


\section{Concluding remarks}

In this paper, we initiated the study of \textsc{Orientable Burning Number} (OBN), which is the problem of finding an orientation of a graph that maximizes the burning number.
We first observed some graph-theoretic bounds and then showed algorithmic and complexity results.

We showed that OBN is NP-hard even on some classes of sparse graphs (\cref{thm:NP-hardness}).
On the other hand, we do not know whether it belongs to NP\@.
We can see that OBN belongs to $\mathrm{\Sigma}^{\mathrm{P}}_{2}$
since it is an $\exists\forall$-problem that asks for the existence of an orientation of a given graph
such that all short sequences of fires are not burning sequences of the oriented graph (see \cite{Woeginger21} for a friendly introduction to $\mathrm{\Sigma}^{\mathrm{P}}_{2}$).
It would be natural to suspect that the problem is indeed $\mathrm{\Sigma}^{\mathrm{P}}_{2}$-hard.
\begin{question}
Does OBN belong to NP, or is it $\mathrm{\Sigma}^{\mathrm{P}}_{2}$-complete?
\end{question}

In contrast to the NP-hardness of the general case,
we showed that the problem is solvable in polynomial time on bipartite graphs or more generally on {\KE} graphs (\cref{cor:ke-graphs}).
We also showed that for perfect graphs, which form a large superclass of bipartite graphs,
we can compute the orientable burning number with an additive error of~$2$ (\cref{cor:perfect_graphs}).
Given these facts, we would like to ask whether the problem can be solved in polynomial time on perfect graphs or on some of its subclasses such as chordal graphs.
\begin{question}
Is OBN polynomial-time solvable on perfect graphs, or on some of its (non-bipartite) subclasses such as chordal graphs?
\end{question}

In the parameterized setting, we showed that OBN parameterized by the target burning number~$b$ is W[1]-hard in general (\cref{thm:W[1]-hardness}),
while it is fixed-parameter tractable on some sparse graphs such as planar graphs (\cref{cor:fpt_avgdeg+b}).
We then studied the setting where $b$ is not part of the parameter.
In this case, we showed that OBN parameterized solely by vertex cover number (or more generally by cluster vertex deletion number plus clique number)
is fixed-parameter tractable (\cref{thm:cvd+omega}).
It would be interesting to study the complexity of parameterizations by more general parameters,
e.g., vertex integrity~\cite{GimaHKKO22}.
\begin{question}
Is OBN fixed-parameter tractable when parameterized solely by
treewidth, pathwidth, treedepth, vertex integrity, or other related parameters?
\end{question}

Finally, we ask a graph-theoretic question.
Most of the algorithmic and complexity results in this paper directly or indirectly used  the relations
between the orientable burning number and the independence number shown in \cref{sec:bounds}.
As shown there, we have $\alpha(G) \le \obn(G)$ and $\obn(G) \in O(\alpha(G) \cdot \log n)$.
Now the question would be the maximum difference between $\alpha(G)$ and $\obn(G)$.
At this point, we only know that the maximum gap is at least~$2$ as $\obn(K_{n}) = 3 = \alpha(K_{n}) + 2$ for $n \ge 5$.
\begin{question}
Is there a graph $G$ with $\obn(G) > \alpha(G) + 2$?
Is there a function $f$ such that $\obn(G) \le f(\alpha(G))$ for every graph $G$?
\end{question}


%
%
\bibliographystyle{splncs04}
\bibliography{ref}

\end{document}